\documentclass{birkjour}
\usepackage{mathrsfs}
\usepackage{bm}
\usepackage{colordvi}

\usepackage{color}

\newcommand{\rme}{\mathrm{e}}
\newcommand{\rmi}{\mathrm{i}}
\newcommand{\rmd}{\mathrm{d}}
\newcommand{\balpha}{{\bm{\alpha}}}
\newcommand{\bell}{{\bm{\ell}}}
\newcommand{\bA}{{\bf A}}
\newcommand{\OO}{{\mathcal{O}}}
\DeclareMathOperator{\sgn}{sgn}
\DeclareMathOperator{\tr}{tr}

 \newtheorem{thm}{Theorem}[section]
 \newtheorem{cor}[thm]{Corollary}
 
 \newtheorem{prop}[thm]{Proposition}
 \theoremstyle{definition}
 
 \theoremstyle{remark}
 
 \newtheorem{rems}[thm]{Remarks}
 \newtheorem*{ex}{Example}
 \numberwithin{equation}{section}

\begin{document}

%
%
%
%
%
%
%
%
%

\title[Spectral Properties of Magnetic Chain Graphs] {
	Spectral Properties of Magnetic Chain Graphs
}

\author{P.~Exner}

\address{%
Doppler Institute for Mathematical Physics and Applied
Mathematics, \\ Czech Technical University in Prague,
B\v{r}ehov\'{a} 7, 11519 Prague, \\ and  Nuclear Physics Institute
ASCR, 25068 \v{R}e\v{z} near Prague, Czechia}

\email{exner@ujf.cas.cz}

\thanks{The research was supported by the Czech Science Foundation (GA\v{C}R) within the project 14-06818S and by the European Union with the project `Support for research teams on CTU' CZ.1.07/2.3.00/30.0034.}
\author{S.~Manko
}
\address{
Department of Physics, Faculty of Nuclear Science and Physical Engineering, \\ Czech Technical University in Prague,
Pohrani\v{c}n\'{i} 1288/1,  40501 D\v{e}\v{c}\'{i}n,
Czechia}
\email{stepan.manko@gmail.com}


\date{January 1, 2004}

\begin{abstract}
We discuss spectral properties of a charged quantum particle confined to a chain graph consisting of an infinite array of rings under influence of a magnetic field assuming a $\delta$-coupling at the points where the rings touch. We start with the situation when the system has a translational symmetry and analyze spectral consequences of perturbations of various kind, such as a local change of the magnetic field, of the coupling constant, or of a ring circumference. A particular attention is paid to weak perturbations, both local and periodic; for the latter we prove a version of Saxon-Hutner conjecture.
\end{abstract}

\maketitle

\section{Introduction}

Quantum graphs are a subject both mathematically rich and useful in applications; we refer to the recent monograph \cite{BK13} for a broad panorama of this theory. One class of interest in this category are chain graphs having the form of an array of rings to which a quantum particle is confined. At the touching points of rings the wave functions are connected by suitable matching conditions, one of the simplest possibilities being the so-called $\delta$-conditions described below.

As long as such a chain is periodic, the spectrum of the corresponding Hamiltonian has a band structure, however, some of the band may be flat since the unique continuation principle of the usual Schr\"odinger theory in general does not hold for operators on graphs with a nontrivial topology. The picture changes when the periodicity is lost, local perturbation may give rise to eigenvalues in the gap of the unperturbed spectrum associated with localized states, see e.g. \cite{DET08}.

The picture becomes even more intriguing if the particle in question is charged and exposed to a magnetic field, even in the simplest situation when the chain graph $\Gamma$ is planar, the field is perpendicular to the graph plane, and no other potential is present. The Hamiltonian then acts as a one-dimensional magnetic Schr\"odinger operator on each edge. Its domain consists of all functions from the Sobolev space $H^2_\mathrm{loc}(\Gamma)$ and we assume that they obey the $\delta$-coupling at the graph vertices which are characterized by conditions
\begin{equation}
\label{cond:deltaCoupling1}
\psi_i(0) = \psi_j(0) = :\psi(0)\,,
\quad
i,j=1,\dots,n\,,
\quad
\sum_{i=1}^n
\mathscr{D}\psi_i(0)
=\alpha\,\psi(0)\,,
\end{equation}
in a vertex where $n$ edges meet. In our case we will have $n=4$; the meaning of the quasiderivative $\mathscr{D}$ in the formula will be explained in the next section.

In our previous paper \cite{EM15} we discussed properties of a straight magnetic chain and eigenvalues coming from local changes of the coupling constant $\alpha$ in \eqref{cond:deltaCoupling1}. Our aim here is to analyze a broader class of perturbations including local variations of the magnetic field and edge lengths. We shall also discuss weak perturbation, not necessarily of local character. If those are periodic we are going to prove an analogue of the well-known Saxon-Hutner conjecture  \cite[Sec.~III.2.3]{AGHH05}.

Let us briefly describe the contents of the paper. As we have mentioned, the periodic magnetic chain was analyzed in \cite{EM15} and we recall here the results only to the degree needed to make the present paper self-contained. Our main technical tool will be a conversion of the problem into a difference equation. Such a duality trick is well known \cite{Ca97, Ex97, Pa13}, however, we have to work it out for our purposes; this will be done in Sec.~\ref{s:dual}. The results are contained in the next two sections. First we show how the discrete spectrum coming from compactly supported perturbations of various sorts can be found and discuss in detail several examples. Next, in Sec.~\ref{s:weak}, we analyze the weak-perturbation situation, first for compactly supported variations of the magnetic field and coupling constants, then for periodic perturbations.

\section{Preliminaries}
\label{sec:prelim}

Let us start with the unperturbed system which is a ring chain $\Gamma$ of the form sketched in Fig.~\ref{fig:graph}; choosing the units appropriately we may suppose without loss of generality that the circumference of each ring is $2\pi$. We also suppose that the particle is exposed to a magnetic field generated by a vector potential $\bf A$; the field is assumed to be perpendicular to the graph plane and homogeneous\footnote{In fact, the only quantity of importance will be in the following the magnetic flux through the rings, hence in general the field must be just invariant with respect to the group of discrete shifts along the chain.}. The corresponding vector potential can be thus chosen tangential to each ring and constant; since the coordinates we use to parametrize $\Gamma$ refer to different orientations in the upper and lower part of the chain, respectively, we choose $-A$ as the potential value on the upper halfcircles and $A$ on the lower ones.

\begin{figure}[h!]
\centering
    \includegraphics[scale=0.8]{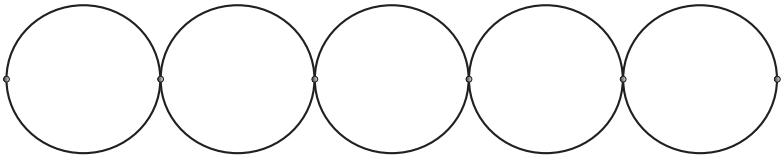}
\caption{The chain graph $\Gamma$}\label{fig:graph}
\end{figure}

Furthermore, in what follows we exclude integer values of $A$ from the consideration. This is possible due to the fact that the spectral properties of the system are invariant with respect to the change of $A$ by an integer which reflects the existence of a simple gauge transformation between such cases. We note that in the chosen units the magnetic flux quantum is $2\pi$ and $A=\frac12 B= \frac{1}{2\pi}\Phi$; we can then rephrase the above claim saying the systems differing by an integer number of flux quanta through each ring are physically equivalent. In this sense the case of an integer $A$ is thus equivalent to the non-magnetic chain treated in \cite{DET08}.

The particle Hamiltonian $-\Delta_{\alpha,A}$ acts as $(-\rmi\nabla-{\bf A})^2$ on each graph edge, with the domain consisting of all functions from the Sobolev space $H^2_\mathrm{loc}(\Gamma)$ which satisfy the boundary conditions (\ref{cond:deltaCoupling1}) at the vertices of $\Gamma$, where the quasiderivatives are equal to the sum of the derivative and the function value multiplied by the tangential component of $e{\bf A}$ as usual \cite{KS03}. The peculiarity of the present model is that we have two pairs of vectors of opposite orientation so their contributions cancel and the left-hand side of (\ref{cond:deltaCoupling1}) is in fact nothing else than the sum of the derivatives taking into account different coordinate orientations, in other words, (\ref{cond:deltaCoupling1}) is the standard $\delta$ coupling. In contrast to \cite{EM15} we assume throughout this paper that the coupling constant $\alpha$ is the same at each vertex  and we are going to determine the band-and-gap structure of the spectrum.

\begin{figure}[h!]
\centering
    \includegraphics[scale=.5]{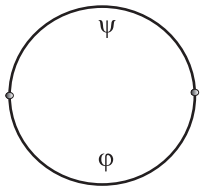}
\caption{Elementary cell of the periodic system}\label{fig:cell}
\end{figure}

In view of the periodicity the band-and-gap structure of the spectrum of $-\Delta_{\alpha,A}$ can be computed using Bloch-Floquet decomposition \cite[Sec.~4.2]{BK13}. Consider an elementary cell with the wave function branches denoted as in Fig.~\ref{fig:cell} and inspect the spectrum of the Floquet components of $-\Delta_{\alpha,A}$. Since the operator acts as $-\mathscr{D}^2:=(\rmi\frac{\rmd}{\rmd x}- A)^2$ on the upper halfcircles and as $-\mathscr{D}^2:=(\rmi\frac{\rmd}{\rmd x}+ A)^2$ on the lower ones, each component of the eigenfunction with energy $E:= k^2\neq0$ is a linear combination of the functions $\rme^{-\rmi Ax}\rme^{\pm\rmi kx}$ on the upper graph links and of $\rme^{\rmi Ax}\rme^{\pm\rmi kx}$ on the lower ones. In what follows we conventionally employ the principal branch of the square root, that is, the momentum $k$ is positive for $E>0$, while for $E$ negative we put $k=\rmi\kappa$ with $\kappa>0$. For a given $E\neq0$, the wave function components on the elementary cell are therefore given by
\begin{alignat}{2}
\label{wavefunct}
\begin{aligned}
\psi(x)&=\rme^{-\rmi Ax}
\big(C^+\rme^{\rmi k  x}+C^-\rme^{-\rmi k  x}\big)\,,
    &\qquad x&\in[0,\pi)\,, \\
\varphi(x)&=\rme^{\rmi Ax}
\big(D^+\rme^{\rmi k  x}+D^-\rme^{-\rmi k  x}\big)\,,
    &\qquad x&\in[0,\pi)\,.
\end{aligned}
\end{alignat}
As we have said, at the contact point the $\delta$-coupling is assumed, i.e. we have
\begin{gather} \label{deltacoupl}
\begin{gathered}
\psi(-0)=\psi(+0)=\varphi(-0)=\varphi(+0)\,,\\
-\mathscr{D}\psi(-0)+\mathscr{D}\psi(+0)-\mathscr{D}\varphi(-0)
+\mathscr{D}\varphi(+0)=\alpha\,\psi(0)\,,
\end{gathered}
\end{gather}
where the left limits, $\psi(-0)$ etc., refer to values on the left-neighbor circle parametrized by $x\in[-\pi,0)$. On the other hand, the function values at the `endpoints' of the cell are related the Floquet conditions,
\begin{align}\label{floquet}
\begin{aligned}
\psi(0)=\rme^{\rmi\theta}\psi(\pi)\,,
\qquad
\mathscr{D}\psi(-0)
=
\rme^{\rmi\theta}\mathscr{D}\psi(\pi-0)\,,\\
\varphi(0)=\rme^{\rmi\theta}\varphi(\pi)\,,
\qquad
\mathscr{D}\varphi(-0)=\rme^{\rmi\theta}\mathscr{D}\varphi(\pi-0)\,,
\end{aligned}
\end{align}
with $\theta$ running through $[-\pi,\pi)$; alternatively we may say that the quasimomentum $\frac1{2\pi}\theta$ runs through $[-\frac12,\frac12)$, the Brillouin zone of the problem. In both cases the vector potential contributions subtract and $\mathscr{D}$ in \eqref{deltacoupl} and \eqref{floquet} can be replaced by the usual derivative.

Using the explicit structure of the wave function  \eqref{wavefunct}, its continuity at the graph vertices as well as Floquet conditions \eqref{floquet}, one easily deduces that
\begin{equation}\label{C+viaC-}
	C^+=\mu(A)C^-\,,\qquad	D^+=\mu(-A)D^-\,,
\end{equation}
where
\[
\mu(A)=-\frac{1-\rme^{\rmi(\theta-A\pi-k\pi)}}
{1-\rme^{\rmi(\theta-A\pi+k\pi)}}\,.
\]
Thus from the first line in \eqref{deltacoupl} we conclude that
\begin{equation}\label{DviaC}
	D^-=\frac{1+\mu(A)}{1+\mu(-A)}C^-\,,
\end{equation}
while from the second one we find after simple manipulations
the characteristic equation for $-\Delta_{\alpha,A}$, namely
\begin{equation}\label{squareEq}
\sin k\pi\cos A\pi
\big(\rme^{2\rmi\theta}-2\xi(k)\rme^{\rmi\theta}+1 \big)=0
\end{equation}
with
\[
\xi(k) := \frac{1}{\cos A\pi}
\Big(\cos k\pi + \frac{\alpha}{4k}\sin k\pi\Big)\,.
\]
The special cases $A-\frac12\in\mathbb{Z}$ and $k \in\mathbb{N}$ lead to infinitely degenerate eigenvalues. Assume that none of this possibilities occur, then we get a quadratic equation for the phase factor $\rme^{\rmi\theta}$ having real coefficients for any $k \in\mathbb{R}\cup\rmi\mathbb{R}\setminus\{0\}$ and the discriminant
\[
D=4(\xi(k)^2-1)\,.
\]
One has to determine the values of $k^2$ for which there is a $\theta\in[-\pi,\pi)$ such that \eqref{squareEq} is satisfied, in other words, those $k^2$ for which it has, as an equation for the unknown
$\rme^{\rmi\theta}$, at least one root of modulus one. Note that a pair of solutions of \eqref{squareEq} always gives one when multiplied, regardless the value of $k$, hence either both roots are complex conjugated numbers of modulus one, or one is of modulus greater than one and the other has the modulus smaller than one. Obviously, the latter situation corresponds to a positive discriminant, and the former one to the discriminant less or equal to zero. We shall not repeat the discussion which the reader can find in \cite{EM15} and limit ourselves to quoting its conclusions.

\begin{thm}
\label{thm:spectrumUnperturbed}
If $A-\frac12\in\mathbb{Z}$; then the spectrum of $-\Delta_{\alpha,A}$ consists of two series of infinitely degenerate eigenvalues, namely $\{k^2\in\mathbb{R}\colon\, \xi(k)=0\}$ and $\{k^2\in\mathbb{R}\colon\, k\in\mathbb{N}\}$.
If $k^2\in\mathbb{R}$ is such that $\xi(k)=0$, then the corresponding $j$th eigenfunction is given by
\begin{align*}
\psi_j(x)&=-\rme^{-\rmi Ax}\sin kx\,,
&\quad
\psi_{j+1}(x)&=\rme^{\rmi A(\pi-x)}\sin k(\pi-x)\,,
\quad
x\in[0,\pi]\,,\\
\varphi_j(x)&=\rme^{\rmi Ax}\sin kx\,,
&\quad
\varphi_{j+1}(x)&=\rme^{\rmi A(x-\pi)}\sin k(x-\pi)\,,
\quad x\in[0,\pi]\,,
\end{align*}
on the $j$th and $(j+1)$st circles and vanishes elsewhere for all $j\in\mathbb{Z}$.
If $k^2\in\mathbb{R}$ is such that $k\in\mathbb{N}$, then the corresponding $j$th eigenfunction is given by
\begin{alignat*}{3}
\psi_j(x)&=\rme^{-\rmi Ax} \sin kx\,,
\quad
\psi_{j+1}(x)&&=(-1)^{k+1}\rme^{\rmi A(\pi-x)} \sin kx\,,
\quad
&&x\in[0,\pi]\,,\\
\varphi_j(x)&=-\rme^{\rmi Ax}\sin kx\,,
\quad
\varphi_{j+1}(x)&&=(-1)^ k \rme^{\rmi A(x-\pi)}\sin kx\,,
\quad
&&x\in[0,\pi]\,,
\end{alignat*}
on the $j$th and $j+1$st circles and vanishes elsewhere.

On the other hand, suppose that $A-\frac12\notin\mathbb{Z}$; then the spectrum of $-\Delta_{\alpha,A}$ consists of infinitely degenerate eigenvalues equal to $k^2$ with $k\in\mathbb{N}$ and with the above eigenfunctions, and absolutely continuous spectral bands with the following properties:

Every spectral band except the first one is contained in an interval $(n^2,(n+1)^2)$ with $n\in\mathbb{N}$. The position of the first spectral band depends on $\alpha$, namely, it is included in $(0,1)$ if  $\alpha>4(|\cos A\pi|-1)/\pi$ or it is negative if $\alpha<-4(|\cos A\pi|+1)/\pi$, otherwise the first spectral band contains zero.
\end{thm}

We finish the section with calculating (for each circle) the quantum flux of particle going through upper and lower semicircles, which are obtained from the formul{\ae}
$J_\psi=-\rmi(\psi^*\mathscr{D}\psi-(\mathscr{D}\psi)^*\psi)$
and
$J_\varphi=-\rmi(\varphi^*\mathscr{D}\varphi-(\mathscr{D}\varphi)^*\varphi)$
with wave functions (1).
Direct calculations give us
\begin{align*}
J_\psi=2k\big(|C^+|^2-|C^-|^2\big)\,,
\\
J_\varphi=2k\big(|D^+|^2-|D^-|^2\big)\,.
\end{align*}
Notice that from the Bloch-Floquet theorem it follows that the above probability currents are conserved along the chain. On the other hand, exploiting formul{\ae} \eqref{C+viaC-}, \eqref{DviaC}, one concludes that
\[
\frac{J_\psi}{J_\varphi}
=
\frac{|\mu(A)|^2-1}{|\mu(-A)|^2-1}
\cdot\bigg|\frac{\mu(-A)+1}{\mu(A)+1}\bigg|^2\,.
\]
As a result, in the non-magnetic case the probability current has a mirror symmetry, $J_\psi/J_\varphi=1$. On the other hand, if a nontrivial magnetic field, $A\not\in\mathbb{Z}$, is applied, the transport becomes in general asymmetric -- cf. also Rem.~\ref{rems}(ii) below.

\section{Dual discrete operators}
\label{s:dual}

In this section we consider more general system, namely, we  assume that the chain graph under consideration has circles of different lengths, and the coupling constants as well as the values of the magnetic field could be different on different circles.
Denote by $\balpha=\{\alpha_j\}_{j\in\mathbb{Z}}$  and $\mathbf{A}=\{A_j\}_{j\in\mathbb{Z}}$ arbitrary but fixed sequences of real values for coupling constants and magnetic fields;  both semicircles of the $j$th graph circle could be associated with the interval $\mathcal{I}_j=(x_j,x_{j+1})$ of the length $\ell_j$ and we require that the following condition holds
\begin{equation}\label{cond:l}
\inf_{\ell\in\bm{\ell}}\ell > 0\,,
\end{equation}
where $\bell=\{\ell_j\}_{j\in\mathbb{Z}}$.
Consider the corresponding magnetic Laplacian $-\Delta_{\balpha,\bell,\bA}$ on such a chain graph, that is, the operator acting as $-\mathscr{D}_j^2=(\rmi\frac{\rmd}{\rmd x}\pm A)^2$ on the $j$th graph ring with the domain consisting of those functions from the appropriate Sobolev space that satisfy the $\delta$ boundary conditions at the graph vertices,
\begin{gather}\label{Delta:cont}
\psi_j(x_j)=\varphi_j(x_j)=
\psi_{j-1}(x_j)=\varphi_{j-1}(x_j)\,,
\\\label{Delta:Kirch}
\mathscr{D}_j\psi_j(x_j)+\mathscr{D}_j\varphi_j(x_j)-\mathscr{D}_{j-1}
\psi_{j-1}(x_j)-\mathscr{D}_{j-1}\varphi_{j-1}(x_j)
=\alpha_j\psi_j(x_j)\,;
\end{gather}
the derivatives are naturally meant as one-sided ones. As before the wavefunction component  $\psi_j$ corresponds to the upper halfcircle of the $j$th ring while $\varphi_j$ stands for the lower one. Our aim in this section is to describe a bijective correspondence between the operator $-\Delta_{\balpha,\bell,\mathbf{A}}$ in $L_2(\Gamma)$ and a certain operator in $\ell_2(\mathbb{Z})$. We are going to write a difference equation every bounded (square summable) solution of which gives rise to a bounded (square integrable) solution of the Schr\"{o}dinger equation corresponding to $-\Delta_{\balpha,\bell,\mathbf{A}}$, and vice versa, every bounded (square integrable) solution of the Schr\"{o}dinger equation produces a bounded (square summable) solution of the said difference equation. This connection plays a crucial role in the following sections, where we will study the spectrum of $-\Delta_{\balpha,\bell,\mathbf{A}}$ for specifically chosen sets $\balpha$, $\bell$ and $\mathbf{A}$.

Denote by $\binom{\psi}{\varphi}$ the general solution of the Schr\"{o}dinger equation
\[
(-\Delta_{\balpha,\bell,\mathbf{A}}- k^2)
\binom{\psi(x, k)}{\varphi(x, k)}=0,\qquad \Im  k\geq 0\,,
\]
which can be rewritten componentwise on the upper and lower semicircles as follows
\begin{align}\label{eq:psi&phi}
\begin{aligned}
(-\mathscr{D}_j^2- k^2)\psi(x, k)&=0
\,,\qquad \Im  k\geq 0\,,\quad x\in \mathcal{I}_j\,,
\\[.3em]
(-\mathscr{D}_j^2- k^2)\varphi(x, k)&=0
\,,\qquad \Im  k\geq 0\,,\quad x\in \mathcal{I}_j\,.
\end{aligned}
\end{align}
Recalling that the wavefunction should satisfy the continuity condition~\eqref{Delta:cont}, we see that the solutions $\psi$ and $\varphi$ and their quasiderivatives are given by
\begin{align}\label{eq:psi}
\begin{aligned}
\psi(x, k)&=
\rme^{\rmi A_j(x_j-x)}
\bigg(
\psi(x_j,k)\cos k(x-x_j)+\mathscr{D}_j\psi(x_j+,k)\,\frac{\sin  k(x-x_j)}{ k}
\bigg)
\,,\\
\mathscr{D}_j\psi(x, k)&=
\rme^{\rmi A_j(x_j-x)}
(
-\psi(x_j,k)\, k\sin k(x-x_j)+\mathscr{D}_j\psi(x_j+,k)\cos k(x-x_j)
)
\end{aligned}
\end{align}
and
\begin{align}\label{eq:phi}
\begin{aligned}
\varphi(x, k)&=
\rme^{\rmi A_j(x-x_j)}
\bigg(
\varphi(x_j, k)\cos k(x-x_j)+\mathscr{D}_j\varphi(x_j+, k)
\,\frac{\sin  k(x-x_j)}{ k}
\bigg)
\,,\\
\mathscr{D}_j\varphi(x, k)&=
\rme^{\rmi A_j(x-x_j)}
(
-\varphi(x_j, k)\,  k\sin k(x-x_j)+\mathscr{D}_j\varphi(x_j+, k)\cos  k(x-x_j))
\end{aligned}
\end{align}
for $\Im k\geq0$, $x\in\mathcal{I}_j$. Next, let us introduce the vector
\[
\Psi_j(k,\tau)=
\begin{pmatrix}
\rme^{\tau}\psi(x_j, k)+\rme^{-\tau}\varphi(x_j, k)
\\[0.5em]
\rme^{\tau}\mathscr{D}_{j-1}\psi(x_j-, k)+\rme^{-\tau}\mathscr{D}_{j-1}\varphi(x_j-,k)
\end{pmatrix},
\]
and the matrix
\[
K(k)=
\begin{pmatrix}
\cos  k\ell_j+\frac{\alpha_j}{2k}\sin  k\ell_j&
\frac1k \sin k\ell_j
\\[0.5em]
- k\sin k\ell_j +\frac{\alpha_j}2\cos k\ell_j
&
\cos k\ell_j
\end{pmatrix}\,,
\]
then taking into account condition \eqref{Delta:Kirch}, we conclude that
\[
K(k)\Psi_j(k,0)
=
\Psi_{j+1}(k,\rmi A_j\ell_j)\,,
\quad \Im  k \geq 0,\quad j\in\mathbb{Z}\,.
\]
In a similar spirit we next introduce the vector
\[
\Phi_j( k)=
\begin{pmatrix}
\psi(x_j, k)+\varphi(x_j, k)
\\[0.5em]
\psi(x_{j-1}, k)+\varphi(x_{j-1}, k)
\end{pmatrix}
\]
to obtain, by virtue of \eqref{eq:psi} and \eqref{eq:phi}, the relations
\begin{align*}
\cos(A_{j-1}\ell_{j-1})\Phi_j( k)
&=
M_{j-1}(k)\Psi_j(k,0)
\\
&=
L_{j-1}( k)\Psi_j( k,\rmi A_{j-1}\ell_{j-1})
\,,
\end{align*}
where the matrices are defined as follows
\begin{align*}
M_j( k)
&=
\begin{pmatrix}
\cos (A_j\ell_j) & 0
\\[0.5em]
\cos  k\ell_j & -\frac1 k\sin  k\ell_j
\end{pmatrix}\,,
\\
L_j( k)
&=
\begin{pmatrix}
1&0
\\[0.5em]
\cos (A_j\ell_j)\cos k\ell_j
& -\frac1 k\cos (A_j\ell_j)\sin k\ell_j
\end{pmatrix}\,.
\end{align*}
Finally, we define the matrix
\begin{align*}
N_j(k)
&=
\frac{\cos (A_{j-1}\ell_{j-1})}{\cos (A_j\ell_j)}
L_j(k)K(k)
(M_{j-1}(k))^{-1}
\\
&=
\frac1{\cos (A_{j}\ell_j)}
\begin{pmatrix}
\frac{\alpha}{2k}\sin (k\ell_j)+
\frac{\sin k(\ell_{j-1}+\ell_j)}{\sin (k\ell_{j-1})}
&
-\frac{\sin(k\ell_j)\cos (A_{j-1}\ell_{j-1})}
{\sin(k\ell_{j-1})}
\\[0.5em]
\cos (A_j\ell_j)&0
\end{pmatrix}
\,,
\end{align*}
where $k\in\mathfrak{K}:= \{z\colon\Im z\geq0\wedge  z\notin\mathbb{Z}\}$, to get the relation
\[
{N}_j( k)\Phi_j( k)=\Phi_{j+1}( k)\,,\qquad  k\in\mathfrak{K}\,,
\]
which by continuity of the wavefunction at the graph vertices can be rewritten as
\[
{N}_j( k)
\binom{\psi_j( k)}{\psi_{j-1}( k)}
=
\binom{\psi_{j+1}( k)}{\psi_{j}( k)}\,,\qquad  k\in\mathfrak{K}\,,
\]
or equivalently as
\begin{multline}
\label{discr_relat}
\sin(k\ell_{j-1})\cos (A_j\ell_j)\psi_{j+1}(k)
+
\sin(k\ell_j)\cos (A_{j-1}\ell_{j-1})\psi_{j-1}(k)
\\
=
\Big(
	\frac{\alpha}{2k}\sin(k\ell_{j-1})\sin(k\ell_{j})+
	\sin k(\ell_{j-1}+\ell_{j})
\Big)
\psi_j(k)\,,\qquad  k\in\mathfrak{K}\,,
\end{multline}
where $\psi_j( k):=\psi(x_j, k)$. Now we can summ up the above calculations.
\begin{thm}\label{thm:Duality}
Suppose that $\alpha_j,\ell_j,A_j\in\mathbb{R}$ for $j\in\mathbb{Z}$ and that \eqref{cond:l} holds. Then any solutions $\psi(x,k)$  and $\varphi(x,k)$, $k^2\in\mathbb{R}$, $k\in\mathfrak{K}$, of~\eqref{eq:psi&phi} satisfy relation~\eqref{discr_relat}. Conversely, any solution of difference equation~\eqref{discr_relat} defines via
\begin{align}\label{eq:psi(x)VSpsi(k)}
\begin{aligned}
\psi(x, k)&=
\rme^{\rmi A_j(x_j-x)}
\bigg\{
\psi_j(k)\cos  k(x-x_j)
\\
&\;\;\;\;
+
(\psi_{j+1}(k)\rme^{\rmi A_j\ell_j}-\psi_j(k)\cos k\ell_j)
\frac{\sin  k(x-x_j)}{\sin k\ell_j}
\bigg\}
\,,
\\
\varphi(x, k)&=
\rme^{\rmi A_j(x-x_j)}
\bigg\{
\psi_j(k)\cos  k(x-x_j),
\\
&\;\;\;\;
+
(\psi_{j+1}(k)\rme^{-\rmi A_j\ell_j}-\psi_j(k)\cos k\ell_j)
\frac{\sin  k(x-x_j)}{\sin k\ell_j}
\bigg\}
\,,
\end{aligned}
\end{align}
$x\in\mathcal{I}_j$, solutions of equations~\eqref{eq:psi&phi} satisfying $\delta$-coupling conditions~\eqref{Delta:cont},~\eqref{Delta:Kirch}.
In addition,  $\binom{\psi(\cdot, k)}{\varphi(\cdot, k)}\in L_p(\Gamma)$ if and only if $\{\psi_j(k)\}_{j\in\mathbb{Z}}\in \ell_p(\mathbb{Z})$ for $p\in\{2,\infty\}$.
\end{thm}
\begin{proof}
It remains to prove the last statement (in both directions). Let $ k^2\in\mathbb{R}$, $ k\in\mathfrak{K}$, and assume all the solutions $\psi(\cdot, k)$, $\varphi(\cdot, k)$ and $\psi_j( k)$ to be real. If $\psi,\varphi\in L_p(\mathbb{R})$ and thus $\mathscr{D}_j^2\psi,\mathscr{D}_j^2\varphi\in L_p(\mathbb{R})$ we infer $\mathscr{D}_j\psi,\mathscr{D}_j\varphi\in L_p(\mathbb{R})$ for all $1\leq p\leq\infty$ and $j\in\mathbb{Z}$. Then $\{\psi_j(k)\}_{j\in\mathbb{Z}}\in \ell_p(\mathbb{Z})$ follows from
\[
\psi(x_j, k)=\rme^{\rmi A_j(x-x_j)}
\bigg(
\psi(x, k)\cos k(x-x_j)-\mathscr{D}_j\psi(x, k)\,
\frac{\sin k(x-x_j)} k
\bigg)
\,,
\]
$x\in\mathcal{I}_j$, for $p=\infty$, and from
\begin{align*}
&(\psi(x_j, k))^2+\frac1{ k^2}(\mathscr{D}_j\psi(x_j^+, k))^2
\\
&\quad
=\rme^{2\rmi A_j(x-x_j)}
\bigg(
(\psi(x, k))^2+\frac1{ k^2}
(\mathscr{D}_j\psi(x, k))^2
\bigg)\,,\qquad x\in\mathcal{I}_j\,,
\end{align*}
for $p=2$. Conversely, assume $\{\psi_j( k)\}_{j\in\mathbb{Z}}\in \ell_p(\mathbb{Z})$ holds for $p=\infty$ or $p=2$. The case $p=\infty$ directly results from \eqref{eq:psi(x)VSpsi(k)} and the case $p=2$ follows from \eqref{eq:psi(x)VSpsi(k)} in combination with
\begin{align*}
&\rme^{2\rmi A_j(x-x_j)}
\bigg(
(\psi(x, k))^2+\frac1{ k^2}
(\mathscr{D}_j\psi(x, k))^2
\bigg)
\\
&\quad
=
(\psi(x_j, k))^2+
    \bigg(
        \frac{\psi_{j+1}( k)\rme^{\rmi A_j\pi}-\psi_j( k)\cos k\pi}{\sin k\pi}
    \bigg)^2
\,,\qquad x\in\mathcal{I}_j\,;
\end{align*}
this completes the proof.
\end{proof}

Before proceeding further, let us note that using Theorem~\ref{thm:Duality} one can rewrite formul{\ae} for probability current on the $j$th circle in terms of $\psi_j(k)$, specifically
\begin{align*}
J_\psi(k)&=
\frac{2k}{\sin k\pi}\big\{\Re\psi_j(k)
(\Re\psi_{j+1}(k)\sin A_j\pi
+
\Im\psi_{j+1}(k)\cos A_j\pi)
\\
&\;\;\;\;
-
\Im\psi_j(k)
(\Re\psi_{j+1}(k)\cos A_j\pi
-
\Im\psi_{j+1}(k)\sin A_j\pi)\big\}\,,
\\
J_\varphi(k)&=-\frac{2k}{\sin k\pi}\big\{
\Re\psi_j(k)
(\Re\psi_{j+1}(k)\sin A_j\pi
-
\Im\psi_{j+1}(k)\cos A_j\pi\big)
\\
&\;\;\;\;
-
\Im\psi_j(k)
(\Re\psi_{j+1}(k)\cos A_j\pi
+
\Im\psi_{j+1}(k)\sin A_j\pi)\big\}\,.
\end{align*}

\begin{rems}\label{rems}
\begin{itemize}
\item[(i)] Since we consider the above formul{\ae}  for $k\in\mathfrak{K}$, the corresponding denominators never vanish.
On the other hand, integer values of $k$ lead to a `compact' sequences $\{\psi_j(k)\}_{j\in\mathbb{Z}}$ with a finite number of nonzero elements.

\item[(ii)] In the nonmagnetic case when $A_j\in\mathbb{Z}$, the probability currents on upper and lower edges with the same vertices are equal, namely,
\[
J_\psi(k)=J_\varphi(k)=
\frac{2k\cos A_j\pi}{\sin k\pi}
\big(
\Re\psi_j(k)\Im\psi_{j+1}(k)
-\Re\psi_{j+1}(k)\Im\psi_j(k)\big)\,.
\]
Note that this is related to the symmetry of the $\delta$ interaction involved. If we replace it by an asymmetric coupling, interesting and [
possibly important `switching' patterns between the upper and lower parts of the graph may occur \cite{CP15}. As we have seen in the previous section, a nontrivial magnetic field can lead to another asymmetry.

\item[(iii)] In what follows we will demonstrate that for those $k$ which produce $\ell_2$-sequences $\{\psi_j(k)\}_{j\in\mathbb{Z}}$, the latter can be chosen to be real, whence the probability currents read as follows
\begin{align*}
J_\psi(k)&=2k\psi_j(k)\psi_{j+1}(k)\frac{\sin A_j\pi}{\sin k\pi}\,,
\\
J_\varphi(k)&=-2k\psi_j(k)\psi_{j+1}(k)\frac{\sin  A_j\pi}{\sin k\pi}\,.
\end{align*}
In other words, since the coefficients $\psi_j(k)\psi_{j+1}(k)$ decay for a fixed $k$ as $|j|\to\infty$, the probability current is `circling' around such localized solutions.
\end{itemize}
\end{rems}

\section{Local perturbations of periodic systems}
\label{s:local}

After these preliminaries, let us pass to our proper topic. Let the Hamiltonian of the periodic system $-\Delta_{\alpha,A}$ be defined as in Sec.~\ref{sec:prelim}, and suppose that the system suffers compactly supported perturbations of different types. In particular, we consider systems with modified values of the magnetic field on two adjacent rings, systems with modified values of the coupling constant and the magnetic field on a fixed ring, and finally, systems with one ring of a different length. Our goal here is to relate spectral properties of perturbed and periodic Hamiltonians.

Let us start with the essential spectrum. We note that if we decouple the straight chain by changing the matching condition to separation ones, the essential spectrum will not be affected as the resolvents of corresponding differ a finite-rank perturbation. This also means that the essential spectrum of the halfchains is again the same, just its multiplicity is one instead of two. Since we shall consider perturbations of a compact support, we may cut such a perturbed chain at two points at both sides of perturbation support. Irrespective of the type of the perturbation, the middle part spectrum is discrete, hence since cuts are associated with a finite rank, the essential spectrum is again the same as for the straight chain. Our aim now is to look into the properties of the discrete spectrum of each perturbed operator.

\subsection{Perturbations of the magnetic field}
Let us consider a periodic system with parameters $\balpha=\{\dots,\alpha,\alpha,\dots\}$, $\bell=\{\dots,\pi,\pi,\dots\}$ and $\bA=\{\dots,A,A,\dots\}$, and suppose that on the indicated couple of neighboring rings we modify the values of the magnetic field, say, to $A_1$ and $A_2$. Without loss of generality we may employ a ring numbering starting from the chosen pair labeling it by the indices $j=1,2$. The perturbed Hamiltonian representing the magnetic field $A$ on each ring of the chain except the first two and $A_1,\,A_2$ on the first and second rings, respectively, will be then denoted by $-\Delta_{ A_1, A_2}$. Our goal in this subsection is to find spectral properties of this operator.

According to our circle numeration the matrices $N_j$ from the proof of Theorem~\ref{thm:Duality} take the form
\begin{eqnarray*}
N_1(k)&=
\begin{pmatrix}
2\xi_1(k)
 & -\frac{\cos A\pi}{\cos A_1\pi}
\\[5pt]
1 & 0
\end{pmatrix}
\,,
\quad &\xi_1(k)=\frac{1}{\cos A_1\pi}
\Big(\cos k\pi + \frac{\alpha}{4k}\sin k\pi\Big)\,, \\
N_2(k)&=
\begin{pmatrix}
2\xi_2(k)& -\frac{\cos A_1\pi}{\cos A_2\pi}
\\[2pt]
1 & 0
\end{pmatrix}
\,,
    \quad &\xi_2(k)=\frac{1}{\cos A_2\pi}
\Big(\cos k\pi + \frac{\alpha}{4k}\sin k\pi\Big)\,, \\
N_3(k)&=
\begin{pmatrix}
2\xi(k)& -\frac{\cos B_2\pi}{\cos A\pi}
\\[5pt]
1 & 0
\end{pmatrix}
\,,
    \qquad &
\end{eqnarray*}
while for $j\in\mathbb{Z}\setminus\{1,2,3\}$ one infers that  $N_j(k)=N(k)$, where
\[
N(k):=
\begin{pmatrix}
2\xi(k)& -1
\\[2pt]
1 & 0
\end{pmatrix}
\]
with $\xi(k)$ being the quantity that appeared in \eqref{squareEq}. As before we have the relations
\[
\Phi_{j+1}(k)
=
N_j(k)
\Phi_j(k)
\,,\qquad j\in\mathbb{Z}\,,
\]
from which it follows that
\begin{alignat}{2}
        \label{eq:PhiJViaPhi4}
\Phi_{j+4}(k) &= (N(k))^j\Phi_4(k)\,,
    &\qquad j &\in\mathbb{N}\,, \\
        \label{eq:Phi4ViaPhi1}
\Phi_4(k) &= N_3(k)N_2(k)N_1(k)\Phi_1(k)\,,
    && \\
        \label{eq:PhiJViaPhi1}
\Phi_{j+1}(k)     &= (N(k))^j\Phi_1(k)\,,
    &\qquad -j&\in\mathbb{N}\,.
\end{alignat}
It is clear that the asymptotical behavior of the norms of $\Phi_j$ is determined by spectral properties of the matrix $N$. Specifically, let $\Phi_4$ be an eigenvector of $N$ corresponding to an eigenvalue $\mu$, then $|\mu| < 1$ (or $|\mu|>1$, $|\mu|=1$) means that $\|\Phi_j\|$ decays exponentially with respect to $j>4$ (respectively, it is exponentially growing, or independent of $j$). At the same time if $\Phi_1$ if an eigenvector of $N$ corresponding to an eigenvalue $\mu$ such that $|\mu| > 1$, then $\|\Phi_j\|$ decays exponentially with respect to $j<1$ (with similar conclusions for $|\mu|<1$ and $|\mu|=1$).

By virtue of Theorem~\ref{thm:Duality} the wavefunction components on the $j$-th ring are determined by $\Phi_j$, and thus, in view of \eqref{eq:PhiJViaPhi4} and \eqref{eq:PhiJViaPhi1}, by $\Phi_4$ or $\Phi_1$ depending on the sign of $j$. If $\Phi_4$ has a non-vanishing component related to an eigenvalue of $N$ of modulus larger than $1$, or $\Phi_1$ has a non-vanishing component related to an eigenvalue of modulus less than $1$, then the corresponding coefficients $\Phi_j$ determine neither an eigenfunction nor a generalized eigenfunction of $-\Delta_{A_1,A_2}$. On the other hand, if $\Phi_4$ is an eigenvector, or a linear combination of eigenvectors, of the matrix $N$ with modulus less than one (respectively, equal to one), and at the same time $\Phi_1$ is an eigenvector, or a linear combination of eigenvectors, of the matrix $N$ with modulus larger than one (respectively, equal to one), then the coefficients $\Phi_j$ determine an eigenfunction (respectively, a generalized eigenfunction) and the corresponding energy $E$ belongs to the point (respectively, continuous) spectrum of the operator $-\Delta_{A_1,A_2}$. To perform the spectral analysis of $N(k)$, we employ its characteristic polynomial at energy $k^2$,
\begin{equation}\label{charPoly}
\lambda^2-2\xi(k)\lambda+1;
\end{equation}
it shows that $N(k)$ has an eigenvalue of modulus less than one \emph{iff} the discriminant of \eqref{charPoly} is positive, i.e.
\[
|\xi(k)|>1\,,
\]
and a pair of complex conjugated eigenvalues of modulus one \emph{iff} the above quantity is less than or equal to one. In the former case the eigenvalues of $N(k)$ are given by
\begin{equation}\label{eq:eigenvalues}
\lambda_{1,2}(k)=
\xi(k)
\pm\sqrt{\xi(k)^2-1},
\end{equation}
satisfying $\lambda_2=\lambda_1^{-1}$, hence $|\lambda_2|<1$ holds if $\xi(k)>1$ and $|\lambda_1|<1$ if this quantity is $<-1$. Moreover, the corresponding eigenvectors of $N(k)$ are
\begin{equation}\label{eigenFunc}
    u_{1,2}(k) = \binom{1}{\lambda_{2,1}(k)}\,.
\end{equation}
It is convenient to abbreviate
\begin{align*}
\lambda_*(k)&=\xi(k)-\sgn(\xi(k))\sqrt{\xi(k)^2-1}\,,
\\
\lambda^*(k)&=\xi(k)+\sgn(\xi(k))\sqrt{\xi(k)^2-1}\,.
\end{align*}
The above functions coincide with $\lambda_1(k)$ and $\lambda_2(k)$ if $\xi(k)<-1$ or with $\lambda_2(k)$ and $\lambda_1(k)$ if $\xi(k)>1$, hence the
absolute value of $\lambda^*(k)$ is greater than 1 while the absolute value of $\lambda_*(k)$ is less than 1 unless $k^2\in\sigma(-\Delta_{\alpha,A})$. We also introduce the vectors
\[
u_*(k)=
\binom{1}{\lambda^*(k)}
\quad\text{and}\quad
u^*=
\binom{1}{\lambda_*(k)}\,,
\]
which play the role of the corresponding eigenvectors.

\begin{prop}
Assume that $k^2\in\mathbb{R}\setminus\sigma(-\Delta_{\alpha,A})$; then $k^2$ is an eigenvalue of $-\Delta_{A_1,A_2}$
iff the relation
\begin{equation}\label{cond:evExist}
\xi(k)\lambda^*(k)=
\frac{(\cos A_1\pi)^2+(\cos A_2\pi)^2}{2(\cos A\pi)^2}
\,.
\end{equation}
is valid for this $k$.
\end{prop}
\begin{proof}
Observe that by virtue of Theorem~\ref{thm:Duality} and formul{\ae}~\eqref{eq:PhiJViaPhi4} and \eqref{eq:PhiJViaPhi1} the only possibility to construct an eigenfunction of $-\Delta_{A_1,A_2}$ is by demanding that
\[
\Phi_4\sim u_*
\quad\text{and}\quad
\Phi_1\sim u^*\,.
\]
We conclude from relation~\eqref{eq:Phi4ViaPhi1} that $k^2$ is an eigenvalue of $-\Delta_{A_1,A_2}$ \emph{iff}
\[
\det[N_3(k)N_2(k)N_1(k) u^*(k),u_*(k)]=0\,,
\]
where the symbol $[a,b]$ here and in the following is used for a $2\times2$ matrix with the columns $a$ and $b$. Taking into account the explicit structure of the matrices $N_j$, we obtain that
\[
N_3N_2N_1=
\begin{pmatrix}
8\xi\xi_1\xi_2
-
2\xi
\Big(
\frac{\cos B_1\pi}{\cos B_2\pi}
+
\frac{\cos B_2\pi}{\cos B_1\pi}
\Big)
&-4\xi_1\xi_2+\frac{\cos B_2\pi}{\cos B_1\pi}
\\[5pt]
4\xi_1\xi_2-\frac{\cos B_1\pi}{\cos B_2\pi}
&-2\xi_1\frac{\cos A\pi}{\cos B_2\pi}
\end{pmatrix}.
\]
Exploiting the identities $\xi_j=\xi\frac{\cos A\pi}{\cos A_j\pi}$, $(\lambda^*)^2=2\xi\lambda^*-1$ and the fact that $|\lambda^*|>1$, one can easily rewrite the above determinant in the form of \eqref{cond:evExist}, completing thus the proof of the proposition.
\end{proof}

Recall that we consider those $k^2$ from the real line that are not in the spectrum of $-\Delta_{\alpha,A}$. For such values of the energy the function $\xi\lambda^*$ is greater than one, hence the equation~\eqref{cond:evExist} admits no solution unless
\begin{equation}\label{NeccSuff}
\frac{(\cos A_1\pi)^2 + (\cos A_2\pi)^2}{2(\cos A\pi)^2} > 1\,.
\end{equation}
Let us note that condition~\eqref{NeccSuff} is also sufficient for eigenvalue existence. In fact, if it holds, then the right-hand side of~\eqref{cond:evExist} takes values in $(1,\frac{1}{(\cos A\pi)^2}]$. Exploiting the explicit structure of the functions $\xi$, we see that the interval $(1,\frac{1}{(\cos A\pi)^2}]$ is in the range of $\xi\lambda^*$ on every spectral gap, since
\[
(\xi\lambda^*)(n)=\frac{1}{(\cos A\pi)^2}+
\frac{1}{|\cos A\pi|}
\sqrt{\frac{1}{(\cos A\pi)^2}-1}
>
\frac{1}{(\cos A\pi)^2}\,.
\]
Thus we get the following claim.
\begin{thm}
The spectrum of $-\Delta_{A_1,A_2}$ coincides with the one for $-\Delta_{\alpha,A}$ unless condition~\eqref{NeccSuff} holds. On the other hand, if~\eqref{NeccSuff} holds, then the essential spectrum of $-\Delta_{A_1,A_2}$ coincides with $\sigma(-\Delta_{\alpha,A})$, and moreover, $-\Delta_{A_1,A_2}$ has precisely one simple eigenvalue in every gap of its essential spectrum.
\end{thm}

As a direct consequence of the theorem is the following
\begin{cor}
Suppose that some $A_j$, say, $A_2$, equals $A$. The spectrum of $-\Delta_{A_1}:=-\Delta_{A_1,A}$ coincides with the one for $-\Delta_{\alpha,A}$ unless the following condition holds
\[
\frac{|\cos A_1\pi|}{|\cos A\pi|}>1\,.
\]
On the other hand, if this is the case, then the essential spectrum of $-\Delta_{A_1}$ coincides with $\sigma(-\Delta_{\alpha,A})$ and, moreover, $-\Delta_{A_1}$ has precisely one simple eigenvalue in every gap of its essential spectrum.
\end{cor}

\subsection{Mixed perturbations}
We again start with a periodic system and suppose that on a certain circle one has a new value of the magnetic field as well as the `left' coupling constant, denoted by $A_1$ and $\alpha_1$ respectively; we introduce ring numbering starting from the chosen one. The Hamiltonian representing this mixed perturbation will be denoted by $-\Delta_{\alpha_1, A_1}$ and our goal in this subsection is to demonstrate spectral properties of $-\Delta_{\alpha_1,A_1}$.

As above we employ the recurrent relations $\Phi_{j+1}=N_j\Phi_j$ with the matrices $N_j$ defined as follows
\begin{align*}
N_1(k)&=
\begin{pmatrix}
2\xi_1(k)
 & -\frac{\cos A\pi}{\cos A_1\pi}
\\[5pt]
1 & 0
\end{pmatrix}
\,,
\\
N_2(k)&=
\begin{pmatrix}
2\xi(k)& -\frac{\cos A_1\pi}{\cos A\pi}
\\[5pt]
1 & 0
\end{pmatrix}
\,,
\qquad
N_j(k)=N(k)
\end{align*}
for $j\in\mathbb{Z}\setminus\{1,2\}$. Here
\[
\xi_1(k)=\frac{1}{\cos A_1\pi}
\Big(
	\cos k\pi
	+ \frac{\alpha_1}{4k}\sin k\pi
\Big)\,.
\]
Using the same ideas as in the previous subsection we conclude that the characteristic equation is of the form $\det[N_2(k)N_1(k)u^*(k),u_*(k)]=0$. Thus substituting the explicit expression for the matrices $N_j$ and the function $\xi_1$ we find that $k^2$ is an eigenvalue of $-\Delta_{\alpha_1,A_1}$ if and only if
\begin{equation}
\alpha_1-\alpha=
\frac{2k\cos A\pi}{\sin k\pi}
\bigg(
\lambda_*(k)
\bigg(\frac{\cos A_1\pi}{\cos A\pi}\bigg)^2
-2\sgn(\xi(k))\sqrt{(\xi(k))^2-1}
-\lambda_*(k)
\bigg)\,.
\end{equation}
holds for this $k$. Let $f$ stand for the right-hand side of the above relation. Then, as $k^2$ varies from the lower end of a gap in $-\Delta_{\alpha,A}$ to the upper end, $f(k)$ is continuous with respect to $k$ and strictly increasing with respect to~$k^2$. In particular, if $|\cos A_1\pi|>|\cos A\pi|$, $f(k)$ alternately increases from $-\infty$ to some positive number or from some negative number to $\infty$, starting with the increase from $-\infty$ in the first gap (the one below the continuum spectrum threshold). The sequence of local maxima (that are positive) is increasing with respect to the gap number, at the same time, the sequence of local minima (that are negative) is decreasing. On the other hand, for $|\cos A_1\pi|<|\cos A\pi|$, $f(k)$ alternately increases from $-\infty$ to some negative number or from some positive number to $\infty$, starting with the increase from $-\infty$ in the first gap. In this case the sequence of local maxima (that are negative) is decreasing with respect to the gap number, while the sequence of local minima (that are positive) is increasing.
\begin{thm}
For $|\cos A_1\pi|>|\cos A\pi|$ and $\alpha_1>\alpha$, the operator $-\Delta_{\alpha_1,A_1}$ has precisely one simple eigenvalue in every gap of its essential spectrum, except possibly a finite number of odd gaps. On the other hand, for $|\cos A_1\pi|>|\cos A\pi|$ and $\alpha_1<\alpha$, it has precisely one simple eigenvalue in every gap of its essential spectrum, except possibly a finite number of even gaps. In particular, for sufficiently small $\alpha_1-\alpha$ positive or negative, there is an eigenvalue in every gap. For $|\cos A_1\pi|<|\cos A\pi|$ and $\alpha_1>\alpha$, there is precisely one simple impurity state in a finite number of even gaps; for $|\cos A_1\pi|<|\cos A\pi|$ and $\alpha_1<\alpha$, there is precisely one simple impurity state in a finite number of odd gaps. In particular, for sufficiently small $\alpha_1-\alpha$ positive or negative, the operator $-\Delta_{\alpha_1,A_1}$ has no eigenvalues.
\end{thm}

\subsection{Perturbations of geometry}

In the final subsection of the section we assume that the periodic system presented by $-\Delta_{\alpha,A}$ suffers a geometric perturbation. Specifically, we suppose that the first ring of the chain is rescaled in such a way that his length changes to $2\ell$, the equal distances between the two vertices being preserved, while the other characteristics of the system such as coupling constants and magnetic fields remain the same. Denoting for the sake of brevity the corresponding Hamiltonian by $-\Delta_\ell$, we are going to show that its spectral properties depend, in particular, on whether the scaling transformation is a contraction ($\ell<\pi$) or dilatation ($\ell>\pi$), as well as on the sign of the coupling constant.

Using the same reasoning as in the previous subsections, we first find the matrices $N_j$ entering the recurrence relations $\Phi_{j+1}=N_j\Phi_j$. For the perturbation in question we have
\[
N_1(k)=
\begin{pmatrix}
2\xi_1(k)
 & -\frac{\sin k\ell\cos A\pi}{\sin k\pi\cos A\ell}
\\[5pt]
1 & 0
\end{pmatrix}
\,,
\quad
N_2(k)=
\begin{pmatrix}
2\xi_2(k)& -\frac{\sin k\pi\cos A\ell}{\sin k\ell\cos A\pi}
\\[5pt]
1 & 0
\end{pmatrix}
\,,
\]
with
\begin{align*}
\xi_1(k)&=\frac{1}{\cos A\ell}
\bigg(
	\frac{\sin k(\pi+\ell)}{2\sin k\pi}
	+ \frac{\alpha}{4k}\sin k\ell
\bigg)\,,
\\[.3em]
\xi_2(k)&=\frac{1}{\cos A\pi}
\bigg(
	\frac{\sin k(\pi+\ell)}{2\sin k\ell}
	+ \frac{\alpha}{4k}\sin k\pi
\bigg)\,,
\end{align*}
while for $j\in\mathbb{Z}\setminus\{1,2\}$, we get $N_j(k)=N(k)$. Furthermore, the eigenvalues of $-\Delta_\ell$ are determined by the characteristic equation
\[
\det[N_2(k)N_1(k) u^*(k),u_*(k)]=0\,.
\]
Taking into account the explicit structure of the matrices $N_j$, we find that
\[
N_2(k)N_1(k)=
\begin{pmatrix}
4(\xi_2(k))^2
\frac{\sin k\ell\cos A\pi}{\sin k\pi\cos A\ell}
-
\frac{\sin k\pi\cos A\ell}{\sin k\ell\cos A\pi}
&-2\xi_2(k)
\frac{\sin k\ell\cos A\pi}{\sin k\pi\cos A\ell}
\\[5pt]
2\xi_2(k)\frac{\sin k\ell\cos A\pi}{\sin k\pi\cos A\ell}
&-
\frac{\sin k\ell\cos A\pi}{\sin k\pi\cos A\ell}
\end{pmatrix}\,,
\]
from which we conclude that $k^2\in\mathbb{R}\setminus\sigma(-\Delta_{\alpha,A})$ is an eigenvalue of $-\Delta_\ell$ if and only if
\begin{equation}\label{geo}
\Big|2\xi_1(k)-\lambda_*(k)
\frac{\sin k\ell\cos A\pi}{\sin k\pi\cos A\ell}
\Big|=1\,.
\end{equation}
We restrict ourselves to discussing solutions of equation~\eqref{geo} below the continuum spectrum threshold, i.e. in the first spectral gap of $-\Delta_{\alpha,A}$. As $k^2$ varies from the lower end of a gap to the upper end, the left-hand side of~\eqref{geo} is continuous with respect
to $k$. Suppose first that $\alpha$ is negative, then the left-hand side of~\eqref{geo} is strictly decreasing with respect to $k^2$, and in the case of contraction (dilatation) its minimal value is less (respectively, greater) than one, hence one obtains one (respectively, no) solution to equation~\eqref{geo}. Assume next $\alpha>0$. In the case of a contraction the left-hand side of~\eqref{geo} is strictly decreasing with respect to $k^2$ and its local minimum is greater than one. At the same time, in the case of a dilatation the function first strictly decreases to its local minimum the value of which is less than one, and then it increase to its local maximum. Moreover, for a fixed natural $n$ and a sufficiently large $\ell$ the function on the left-hand side of~\eqref{geo} has exactly $n$ local minima and maxima with the following properties: all the minimum values are zero, while the sequence of maximum values is strictly decreasing and greater than one. This means, in particular, that for such an $\ell$ one obtains $2n$ eigenvalues in the first spectral gaps.

Denote by $\sharp_\ell$ the function counting eigenvalues of $-\Delta_\ell$ in the first spectral gap of its continuous spectrum.
\begin{thm}
For $\alpha>0$ and $\ell\in(0,\pi)$, we have $\sharp_\ell=0$. On the other hand, if $\alpha>0$ and $\ell>\pi$, then $\sharp_\ell\geq 1$, and moreover, $\sharp_\ell\to\infty$ holds as $\ell\to\infty$. At the same time, for $\alpha<0$ and $\ell\in(0,\pi)$ we have $\sharp_\ell=1$. If $\alpha<0$ and $\ell>\pi$, then $\sharp_\ell=0$.
\end{thm}

\section{Weak perturbations of periodic systems}
\label{s:weak}

As before let $-\Delta_{\alpha,A}$ stands for the Hamiltonian of the periodic system described in Sec.~\ref{sec:prelim}. Now we are going to discuss situations when the system suffers some weak perturbations. Specifically, we suppose that the parameters of the system are of the form $A+\varepsilon A_j$ and $\alpha+\varepsilon\alpha_j$ with $\varepsilon\in(0,1)$ and ask about the spectrum of the perturbed Hamiltonian in the asymptotic regime $\varepsilon\to 0$. First we focus on perturbations supported on a compact subdomain of the graph and demonstrate the behavior of the eigenvalues in the spectral gaps of the periodic operator. Next we turn to systems where the perturbation is also periodic and show that in this situation a version of the well-known Saxon-Hutner conjecture is valid.

\subsection{Local perturbations of magnetic fields and coupling constants}

Our aim here is to compare spectral properties of $-\Delta_{\alpha,A}$ with those produced by a weak finite-rank perturbation. To be specific, we suppose that the coupling constant perturbation strength is $\varepsilon\alpha_j,\: j = 1,\dots,n$, at the vertices with the coordinates $\{\pi,\dots,n\pi\}$, and at the same time, the magnetic field suffers a weak perturbation, namely for the ring indices $j$ from $\{1,\dots,n\}$ we have an `additional' magnetic potential $\varepsilon A_j$. The perturbation is controlled by the small parameter $\varepsilon$ and the perturbed Hamiltonian will be denoted by $-\Delta_\varepsilon$. In view of the compact support, the essential spectrum is not affected as one can check using the argument used in the opening of Sec.~\ref{s:local}.

We are going to demonstrate that, as $\varepsilon\to0+$, the presence of the eigenvalue in the gap of $\sigma_\mathrm{ess}(-\Delta_\varepsilon)$ is determined by the signs of $\sum_{j=1}^n\alpha_j$ and $\sum_{j=1}^n A_j$ as well as of $\cot A\pi$. With this aim in mind, we mimic the argument from the previous section which yields the relation
\[
\Phi_{j+1}(k)
=
N_j(k,\varepsilon)
\Phi_j(k)
\,,\qquad j\in\mathbb{Z}\,,
\]
where the matrix $N_j$ takes the form
\begin{equation*}
N_j(k,\varepsilon)
=
\begin{pmatrix}
2\xi_j(k,\varepsilon)
&
-\frac{\cos (A_{j-1}\pi)}
{\cos (A_{j}\pi)}
\\[0.5em]
1&0
\end{pmatrix}
\end{equation*}
with
\[
\xi_j(k,\varepsilon)=\frac{1}{\cos(A+\varepsilon A_j)\pi}
\bigg(
	\cos k\pi+(\alpha+\varepsilon\alpha_j)\frac{\sin k\pi}{4k}
\bigg)\,.
\]
Note that the above matrix admits the following asymptotic expansion
\begin{equation*}
N_j
=N+\varepsilon\pi\tan A\pi M_j+\OO(\varepsilon^2)
\end{equation*}
as $\varepsilon\to0$, where
\[
N(k)=
\begin{pmatrix}
2\xi(k) & -1
\\[5pt]
1 & 0
\end{pmatrix}
\quad\text{and}\quad
M_j(k)=
\begin{pmatrix}
2A_j\xi(k)+\frac{\alpha_j\sin k\pi}{2\pi k\sin A\pi} & A_{j-1}-A_j
\\[5pt]
0 & 0
\end{pmatrix}
\]
for $j\in\mathbb{Z}$, we just put $A_j=0$ for $j\in\mathbb{Z}\setminus\{1,2,\ldots,n\}$. We are going to combine these relations with
\begin{alignat}{2}
        \label{eq:PhiJViaPhi2}
\Phi_{n+2+j}(k) &= (N(k))^j\Phi_{n+2}(k)\,,
    &\qquad j &\in\mathbb{N}\,, \\[.3em]
        \label{eq:Phi2ViaPhi0}
\Phi_{n+2}(k) &= \underbrace{N_{n+1}(k)N_n(k)\ldots N_1(k)}_{\mathscr{N}_n(k)}\Phi_1(k)\,,
    && \\[.3em]
        \label{eq:PhiJViaPhi0}
\Phi_{j+1}(k)     &= (N(k))^j\Phi_1(k)\,,
    &\qquad -j&\in\mathbb{N}\,.
\end{alignat}
It is clear that the asymptotic behavior of the norms of the vectors $\Phi_j$ is determined by spectral properties of the matrix $N$.

Note that in view of Theorem~\ref{thm:Duality} and formul{\ae}~\eqref{eq:PhiJViaPhi2} and \eqref{eq:PhiJViaPhi0}
the only possibility to obtain an $L^2$ eigenfunction of $-\Delta_{\varepsilon}$ is by demanding that
\[
\Phi_{n+2}\sim
\binom{1}{\lambda^*}
\quad\text{and}\quad
\Phi_1\sim
\binom{1}{\lambda_*}\,,
\]
thus we conclude from relation~\eqref{eq:Phi2ViaPhi0} that $k^2$ is an eigenvalue of $-\Delta_{\varepsilon}$ \emph{iff}
\begin{equation}\label{determinant}
\det\bigg[\mathscr{N}_{n}(k)\binom{1}{\lambda_*(k)},\binom{1}{\lambda^*(k)}\bigg]=0\,.
\end{equation}
Now we observe that in the limit $\varepsilon\to0$ the product $\mathscr{N}_n$ behaves as
\[
\mathscr{N}_n(k)=
(N(k))^{n+1}+\varepsilon\pi\tan A\pi
	\sum_{j=0}^{n}(N(k))^jM_{n+1-j}(k)(N(k))^{n-j}+\OO(\varepsilon^2)\,,
\]
and using than condition \eqref{determinant}, one gets
\begin{multline*}
\det[\mathscr{N}_n(k)u^*(k),u_*(k)]=
(\lambda^*(k))^{n+1}
\det[u^*(k),u_*(k)]
\\
+(\lambda^*(k))^n\varepsilon\pi\tan A\pi\sum_{j=1}^{n+1}
\det[M_j(k)u^*(k),u_*(k)]
+\OO(\varepsilon^2)\,.
\end{multline*}
It is easy to see that $\det[u^*,u_*]=\lambda^*-\lambda_*$, and moreover,
\[
\det[M_ju^*,u_*]=\lambda^*
\Big(
\lambda^*A_j+\lambda_*A_{j-1}+\frac{\alpha_j\sin k\pi}{2k\pi\sin A\pi}
\Big)\,,
\]
hence that the characteristic determinant takes asymptotically the form
\[
\lambda^*-\lambda_*+2\varepsilon\xi(k)\pi\tan A\pi
\sum_{j=1}^n
\bigg(
A_j+
\frac{\alpha_j\sin k\pi}{4k\pi\xi(k)\sin A\pi}
\bigg)
+\OO(\varepsilon^2)
\]
Finally, observe that the eigenvalues are given via solution to the equation
\begin{equation}\label{eq:charSmall}
\varepsilon g(k)+\OO(\varepsilon^2)=f(k)\,,
\end{equation}
where
\[
f(k)=-\frac{\cot A\pi}{\pi}\sqrt{1-\frac{1}{\xi(k)^2}}\,,
\quad
g(k)=\sum_{j=1}^n
	A_j+\frac{\sin k\pi}{4k\pi\xi(k)\sin A\pi}
	\sum_{j=1}^n\alpha_j
\,.
\]
Obviously, the sign of function $f$ is determined by the number $\sgn(-\cot A\pi)$. Furthermore, as $k^2$ varies from the lower end of a gap in $\sigma(-\Delta_{\alpha,A})$ to its upper end, $f(k)$ is continuous with respect to $k$ and strictly monotonous with respect to $k^2$, and $f(k)$ tends to zero as $\xi(k)\to1$, the value being attained at the spectrum threshold.

At the same time, the function $g$ has the following properties: $g(k)$ is continuous with respect to $k$ as $k^2$ varies from the left infinity to the right infinity of the real line, outside a countable set of the second order jumps $\{k\in\mathbb{C}_+\colon \xi(k)=0\}$. Furthermore, $g(k)$ is strictly increasing (decreasing) with respect to $k^2$ if $\cot A\pi\sum_{j=1}^n\alpha_j>0$ (respectively, $\cot A\pi\sum_{j=1}^n\alpha_j<0$).
Note that $g(k)\to\sum_{j=1}^n A_j$ holds as $k^2\to-\infty$, as well as that $g(k)=\sum_{j=1}^n A_j$ for any non-negative integer $k$ and that these point are inflection points of the function $g$. Finally, consider some  fixed neighborhood of a $k\in\mathbb{N}$, then the slope of the function $g$ in this neighborhood tends to zero as $\mathbb{N}\ni k\to\infty$.

The listed properties of the functions $f$ and $g$ allow us to state the following result describing eigenvalues in the gaps.
\begin{thm}
Assume that $\cot A\pi\sum_{j=1}^n A_j >0$. If $\sum_{j=1}^n\alpha_j>0$, the operator $-\Delta_\varepsilon$ has no eigenvalues as $\varepsilon\to0+$ except in a finite number of even gaps, where it can have one eigenvalue per gap. Similarly, for $\sum_{j=1}^n\alpha_j<0$ and $-\Delta_\varepsilon$ it has no eigenvalues except in a finite number of odd gaps, where it can have one eigenvalue per gap. In particular, for a sufficiently large $|\sum_{j=1}^n A_j|$, the operator has no eigenvalues, while for a sufficiently small  $|\sum_{j=1}^n A_j|$ and $\sum_{j=1}^n\alpha_j>0$ (respectively, $\sum_{j=1}^n\alpha_j<0$) it has an eigenvalue in the second (respectively, the first) gap.

On the other hand, let $\cot A\pi\sum_{j=1}^n A_j<0$. If $\sum_{j=1}^n\alpha_j>0$, the operator $-\Delta_\varepsilon$ has precisely one simple eigenvalue as $\varepsilon\to0$ in every gap of its essential spectrum except possibly a finite number of odd gaps. If $\sum_{j=1}^n\alpha_j<0$, the operator has precisely one simple eigenvalue in every gap of its essential spectrum except possibly a finite number of even gaps. In particular, for a sufficiently large $|\sum_{j=1}^n A_j|$, the operator has an eigenvalue in every gap, while for a sufficiently small  $|\sum_{j=1}^n A_j|$ and $\sum_{j=1}^n\alpha_j>0$ (respectively, $\sum_{j=1}^n\alpha_j<0$) it has no eigenvalue in the first (respectively, the second) gap.
\end{thm}

The theorem does not cover particular situations when some of the perturbations has zero mean. Let us focus on them and suppose first that  $\sum_{j=1}^n \alpha_j=0$, then the characteristic equation~\eqref{eq:charSmall} reads
\[
-\frac{\cot A\pi}{\pi}\sqrt{1-\frac{1}{\xi(k)^2}}
=\varepsilon\sum_{j=1}^n
A_j+\OO(\varepsilon^2)\,,\qquad \varepsilon\to0\,.
\]
We see that for a sufficiently small positive $\varepsilon$ the necessary and sufficient condition for eigenvalue existence takes the form
\begin{equation}\label{NeccSuff:eps}
\sgn\bigg(\sum_{j=1}^nA_j\bigg)=-\sgn(\cot A\pi)\,.
\end{equation}

\begin{thm}
Suppose that $\sum_{j=1}^n\alpha_j=0$. The spectrum of $-\Delta_{\varepsilon}$ coincides with the one for $-\Delta_{\alpha,A}$ unless condition~\eqref{NeccSuff:eps} holds. If this is the case, then the essential spectrum of $-\Delta_{\varepsilon}$ coincides with $\sigma(-\Delta_{\alpha,A})$, and moreover, $-\Delta_{\varepsilon}$ has precisely one simple eigenvalue in every gap of its essential spectrum.
\end{thm}

\noindent On the other hand, for $\sum_{j=1}^n A_j=0$ the characteristic equation~\eqref{eq:charSmall} reads
\begin{equation}\label{m>2}
\varepsilon\sum_{j=1}^n\alpha_j+\OO(\varepsilon^2)=-\sgn(\xi(k))\frac{4k\cos A\pi}{\sin k\pi}
\sqrt{(\xi(k))^2-1}\,,\qquad\varepsilon\to0\,.
\end{equation}
As $k^2$ varies from the lower end of a gap in $\sigma(-\Delta_{\alpha,A})$ to its upper end, the right-hand side of this relation is continuous with respect to $k$ and strictly increasing with respect to $k^2$. In particular, it alternately increases from $-\infty$ to zero or from zero to $\infty$, starting with the increase from $-\infty$ to zero in the first gap, i.e. the one below the continuous spectrum threshold.
\begin{thm}
Suppose that $\sum_{j=1}^n A_j=0$. For any $\varepsilon\in(0,1)$ the essential spectrum of $-\Delta_{\varepsilon}$ coincides with that of $-\Delta_{\alpha,A}$. Assume that $\sum_{j=1}^n\alpha_j<0$, then the operator $-\Delta_{\varepsilon}$ has exactly one simple impurity state in every odd gap of its essential spectrum for $\varepsilon\to0+$. If the sum $\sum_{j=1}^n\alpha_j$ is positive, then it has exactly one simple impurity state in every even gap of its essential spectrum as $\varepsilon\to0$.
\end{thm}

\subsection{Weak periodic perturbations}

Let us turn now to perturbations which are periodic along the chain graph. Let $\bA$ and $\balpha$ be sequences with a period $p\in\mathbb{N}$, i.e.
\begin{alignat*}{3}
\bA && =\{A_j^{(p)}\}_{j\in\mathbb{Z}}\,,
	&& \qquad A_{j+p}^{(p)}&=A_j^{(p)}\,,
		\qquad j\in\mathbb{Z}\,,
\\
\balpha && =\{\alpha^{(p)}\}_{j\in\mathbb{Z}}\,,
		&&\qquad \alpha_{j+p}^{(p)}&=\alpha_j^{(p)}\,,
			\qquad j\in\mathbb{Z}\,,
\end{alignat*}
and ask about the effect of the corresponding perturbations on the spectrum. Then difference equation \eqref{discr_relat} now becomes
\[
\psi_{j+1}(k)
+
\frac{\cos (A^{(p)}_{j-1}\pi)}{\cos (A^{(p)}_j\pi)}\psi_{j-1}(k)
=
2\xi^{(p)}_j(k)
\psi_j(k)\,,
\]
with
\[
\xi_j^{(p)}( k)=
\frac{1}{\cos (A^{(p)}_j\pi)}
\bigg(\cos k\pi+
\frac{\alpha_j^{(p)}}{4 k}
\sin  k\pi\bigg)\,,
\qquad k\in\mathfrak{K}\,,\quad j\in\mathbb{Z}\,.
\]
The relation between the vector whose position indices differing by $p$ can be written as
\[
\prod_{j=1}^p
    N_j^{(p)}(k)\Phi_1(k)
    =
    \Phi_{p+1}(k)\,,\qquad k\in\mathfrak{K}\,,
\]
where $\Phi_j(k)$ has been described in the previous section, while $N_j^{(p)}(k)$ is defined in the same manner as $N_j(k)$ replacing $\xi_j$ with $\xi_j^{(p)}$. Since the determinant of $N_j^{(p)}$ is $\cos (A^{(p)}_{j-1}\pi)/\cos (A^{(p)}_{j}\pi)$, the matrix $\prod_{j=1}^pN_j^{(p)}$ has unit determinant, and therefore the product of its eigenvalues is one. By the Floquet-Bloch theorem $\Phi_{p+1}(k)=\rme^{\rmi p\pi\theta}\Phi_1(k)$ and the eigenvalues are given by $\rme^{\pm\rmi p\pi\theta}$ implying
\begin{equation}\label{eq:trace}
\mathrm{tr}\,
    \bigg(
        \prod_{j=1}^pN_j^{(p)}(k)
    \bigg)
=
    2\cos p\pi\theta,\qquad
\Im k\geq 0\,,\quad \Im\theta\geq0\,.
\end{equation}

\begin{ex}
The original unperturbed system is, of course, included. Indeed, suppose that $p=1$, $A_j^{(p)}=A$ and $\alpha^{(p)}_j=\alpha$, then the characteristic determinant of the corresponding operator reads as
\[
\cos\theta\pi=\xi(k)\,,
\qquad
	k\in\mathfrak{K}\,,
\qquad
\theta\in[-1,1)\,,
\]
which is what we get from the condition \eqref{squareEq} unless $A-\frac12\in\mathbb{Z}$ or $k \in\mathbb{N}$.
\end{ex}

In what follows we will assume that the perturbation is weak and put
\begin{align*}
A_j^{(p)}&=A_j^{(p)}(\varepsilon)=A+\varepsilon A_j\,,
	&	\hspace{-3em}A_{j+p}&=A_j\,,
	&	\hspace{-3em}j&\in\mathbb{Z}\,,\\
\alpha_j^{(p)}&=\alpha_j^{(p)}(\varepsilon)=\alpha+\varepsilon\alpha_j\,,
	&	\hspace{-3em}\alpha_{j+p}&=\alpha_j\,,
	&	\hspace{-3em}j&\in\mathbb{Z}\,,
\end{align*}
denoting by $-\Delta_{\varepsilon}^{(p)}$ the corresponding Hamiltonian. Recall from the previous section that in this case we have
\[
\mathscr{N}_p=N^p+\varepsilon\pi\tan A\pi\sum_{j=0}^{m-1}
N^jM_{p-j}N^{p-j-1}+\OO(\varepsilon^2)\,,
\]
which yields, in particular,
\[
\tr \mathscr{N}_p(k)=\tr(N(k)^p)+\varepsilon\pi\tan A\pi\sum_{j=1}^p
\tr(N(k)^{p-1}M_j)+\OO(\varepsilon^2)
\]
Given the definition of $N(k)$ it is straightforward to check that
\[
N(k)^p=
\begin{pmatrix}
U_p(\xi(k)) & -U_{p-1}(\xi(k))
\\[0.5em]
U_{p-1}(\xi(k)) & -U_{p-2}(\xi(k))
\end{pmatrix}\,,
\]
where $U_p$ are Chebyshev's polynomials of the second kind. Denoting conventionally by $T_p$ Chebyshev's polynomials of the first kind, and using the relation $2T_p(\xi)=U_p(\xi)-U_{p-2}(\xi)$ we arrive at
\[
\tr \mathscr{N}_p(k)=2\big(T_p(\xi(k))+\varepsilon\pi\tan (A\pi)
\xi(k)U_{p-1}(\xi(k))g(k)\big)\,.
\]
to state the result, denote by $-\Delta_{\varepsilon,j}$ the Hamiltonian of the `Kronig-Penney-type' chain graph, $p=1$, with the parameters $\alpha+\varepsilon\alpha_j$ and $A+\varepsilon A_j$ and with the corresponding resolvent set $\rho(-\Delta_{\varepsilon,j})$.

\begin{thm}
Assume that
\[
	\rho_p= \mathbb{R} \cap \bigcap_{j=1}^p
	\rho(-\Delta_{\varepsilon,j})
\]
is the intersection of all the spectral gaps of the Hamiltonians of the Kronig-Penney-type chain graphs indicated above. Then
\[
	\rho_p\subseteq \rho(-\Delta_{\varepsilon}^{(p)})\,.
\]
holds for all $\varepsilon$ small enough.
\end{thm}
\begin{proof}
To begin with observe that energies $k^2\in\mathbb{R}$ in the spectral gaps of all `Kronig-Penney' Hamiltonians $-\Delta_{\varepsilon,j}$ are now simply characterized by
\[
\big|\xi_j^{(p)}(k)\big|>1\,,
\]
which is obviously equivalent to
\begin{equation}\label{pureCrystal}
\big|\xi(k)\big(1+\varepsilon\pi\tan(A\pi)g_j(k)\big)\big|>1 +\OO(\varepsilon^2)\,,\qquad \varepsilon\to0\,,\quad j\in\mathbb{Z}\,,
\end{equation}
where
\[
g_j(k)=
	A_j+\frac{\alpha_j\sin k\pi}{4k\pi\xi(k)\sin A\pi}\,.
\]
Thus in order to prove the theorem it suffices to show that the validity of \eqref{pureCrystal} for a fixed $k$, $\Im k \geq0$, and all $j=1,\dots,p$, implies
\[
|\tr\mathscr{N}_p(k)|>2\,,
\]
or equivalently
\[
\big|
	T_p(\xi(k))
	+
	\varepsilon\pi\tan (A\pi)
		\xi(k)U_{p-1}(\xi(k))g(k)
\big|
>
1+\OO(\varepsilon^2)\,,
\qquad \varepsilon\to0\,.
\]
To this end we will employ explicit expressions of Chebyshev's polynomials of the first and second kind, namely
\[
T_p(\xi)=\frac{1}{2}\sum_{i=0}^{[\frac{p}{2}]}
\frac{p}{p-i}\binom{p-i}{i}(-1)^i(2\xi)^{p-2i}\,,
\]
and
\[
U_p(\xi)=\sum_{i=0}^{[\frac{p}{2}]}
\binom{p-i}{i}(-1)^i(2\xi)^{p-2i}\,.
\]
From them it follows that the characteristic determinant of the total Hamiltonian $-\Delta_{\varepsilon}^{(p)}$ can be rewritten as
\begin{align*}
&T_p(\xi)+\varepsilon\pi\tan (A\pi) \xi U_{p-1}(\xi)g
\\
&\qquad
=
\sum_{i=0}^{\left[\frac{p-1}{2}\right]}
(-1)^i2^{p-2i-1}a(i)b_\varepsilon(i,\xi) +  ((p-1)\;\mathrm{mod}\; 2)\,,
\end{align*}
where
\[
a(i):=\frac{(2i)!(p-i-1)!}{i!(p-1)!}\,.
\]
and
\[
b_\varepsilon(i,\xi):=\bigg(\binom{p}{p-2i}
+\binom{p-1}{p-2i-1}\varepsilon\pi\tan(A\pi) g\bigg)
\xi^{p-2i}\,.
\]
It is worth mentioning here that the factor $a$ is equal to one for $i=0$ and is less than or equal to two for $i\leq[\frac{p}{2}]$.
At the same time the factor $b_\varepsilon$ can be cast into the form
\[
b_\varepsilon(i,\xi)=\sum_{j_1<\dots<j_{p-2i}}
\prod_{n=1}^{p-2i}
\xi(1+\varepsilon\pi\tan (A\pi) g_{j_n})+\OO(\varepsilon^2)\,,
\qquad \varepsilon\to0\,.
\]
From what has been already said we finally conclude that
\begin{align*}
&
\big|
	T_p(\xi(k))+
	\varepsilon\pi\tan (A\pi)\xi(k)U_{p-1}(\xi(k))g(k)
\big|
\\
&\qquad\geq
	2^{p-1}\prod_{j=1}^p|\xi(k)(1+\varepsilon\pi\tan (A\pi) g_j(k))|
\\
&\qquad\quad
	-
	2^{p-3}a(1)\sum_{j_1<\dots<j_{p-2}}
\prod_{n=1}^{p-2}
|\xi(k)(1+\varepsilon\pi\tan(A\pi)g_{j_n}(k))|
\\
&\qquad\quad
-\ldots-((p-1)\;\mathrm{mod}\;2)\sum_{j=1}^p
|\xi(k)(1+\varepsilon\pi\tan(A\pi)g_j(k))|
+\OO(\varepsilon^2)
\\
&
\qquad
>1+\OO(\varepsilon^2)\,,\qquad \varepsilon\to0\,.
\end{align*}
By that, the proof of the theorem is complete.
\end{proof}

\noindent One is naturally interested whether the result remains valid also beyond the weak-coupling regime in analogy with Theorem~III.2.3.6. of \cite{AGHH05}. At present, this question remains open.


\end{document}